\documentclass{article}

\newif\ifarxiv
\arxivfalse

\ifarxiv
  \usepackage[preprint]{neurips_2026}
\else
  \usepackage[dblblindworkshop,final]{neurips_2026}
\fi
\workshoptitle{AXIOM Foundations of Efficient Deep Learning}

\usepackage[utf8]{inputenc}
\usepackage[T1]{fontenc}
\usepackage{amsmath,amssymb,amsthm}
\usepackage{booktabs}
\usepackage{float}
\usepackage{microtype}
\usepackage{url}
\usepackage{xcolor}
\usepackage{wrapfig}
\usepackage{tikz}
\usepackage{pgfplots}
\pgfplotsset{compat=1.18}
\usetikzlibrary{calc,positioning}
\usepackage{hyperref}
\hypersetup{
  hidelinks,
  pdfauthor={Harshit Verma, Rex Ying},
  pdftitle={Component-Weighted Centroid Search for Exact Incremental BPE}
}

\newtheorem{theorem}{Theorem}
\newtheorem{lemma}{Lemma}
\newtheorem{proposition}{Proposition}

\newcommand{\pre}{\operatorname{pre}}
\newcommand{\suc}{\operatorname{suc}}
\newcommand{\lastt}{\operatorname{last}}
\newcommand{\dfs}{\operatorname{dfs}}
\newcommand{\bin}{\operatorname{bin}}
\newcommand{\tk}[1]{\texttt{#1}}
\newcommand{\rul}[2]{\tk{#1}\ensuremath{\,\cdot\,}\tk{#2}}

\definecolor{countbal}{HTML}{2A78D6}
\definecolor{compw}{HTML}{EB6834}
\definecolor{muted}{HTML}{77766F}

\newcommand{\codeurl}{\url{https://github.com/harshit5674/component-weighted-bpe}}

\title{Component-Weighted Centroid Search for Exact Incremental BPE}
\author{%
  Harshit Verma \\
  Yale University \\
  \texttt{harshit.verma@yale.edu}
  \And
  Rex Ying \\
  Yale University \\
  \texttt{rex.ying@yale.edu}
}

\begin{document}
\maketitle

\begin{abstract}
Exact incremental tokenizers keep the tokenization of a growing text up to date, so a long prompt or agent trace is not retokenized each time it grows.  Jiang and Gong's algorithm does this for byte-pair encoding (BPE) in $O(\log^2 t)$ worst-case time per byte, where $t$ is the length of the longest token.  Its search passes through $O(\log t)$ nested subproblems and runs a binary search over the possible branches at each one.  We weight each branch by the size of the subproblem it leads to.  Each step then costs about one probe for every halving of the remaining candidates, and a whole update costs $O(\log t)$, as much as one binary search over all candidates.  Output and asymptotic memory are unchanged.  We also construct valid BPE vocabularies on which the original search provably needs $\Theta(\log^2 t)$ probes and ours $\Theta(\log t)$, and a Rust implementation reproduces both counts exactly.  On six real vocabularies the two searches run at about the same speed on ordinary text, while on long runs of spaces or hyphens the weighted search is up to 2.5$\times$ faster.
\end{abstract}

\section{Introduction}

Byte-pair encoding (BPE)~\citep{sennrich2016bpe} tokenizes text by repeatedly merging adjacent symbols according to a ranked list of merge rules.  In chat and agent serving, text grows by appending, and the tokenization of the whole text is needed after each addition.  Retokenizing everything repeats work, while tokenizing only the new piece is not exact, because merges can cross the boundary.  An \emph{exact incremental} tokenizer instead updates its state after every appended byte.

\citet{jiang2026incremental} give such an algorithm.  After each byte it computes the last token of the current tokenization.  Removing the last token of a BPE tokenization leaves the tokenization of the remaining prefix~\citep{berglund2023formalizing}, so these last tokens represent the tokenization of every prefix.  An update takes $O(\log^2 t)$ time in the worst case, where $t$ is the maximum token length.  The new last token is one of at most $t$ candidates.  They form a tree, and the answer is the end of a path that starts at its root (Section~\ref{sec:background}).  A \emph{centroid search} finds this end in $O(\log t)$ steps, each of which discards at least half of the candidates.  A step may also have to choose, by binary search, which of up to $\Theta(t)$ children continues the path, and this costs $\Theta(\log t)$.

Binary search treats all children alike, although one child may lead to half of the remaining candidates and another to a single one.  We weight each child by the size of the subproblem it leads to and search the children with a weight-balanced tree~\citep{mehlhorn1975}.  A child of weight $w$ among children of total weight $W$ is then found in about $\log_2(W/w)$ probes, roughly one probe for each halving of the candidates.  Since $t$ candidates can be halved only $\log_2 t$ times, a whole update costs $O(\log t)$, the same as one binary search over all candidates (Figure~\ref{fig:example}).  Our contributions are:
\begin{itemize}
\item a drop-in replacement for the branch search that lowers the worst-case time per byte from $O(\log^2 t)$ to $O(\log t)$, with the same output and asymptotic memory (Section~\ref{sec:method});
\item valid BPE vocabularies, also over a fixed alphabet, on which the original search makes $\Theta(\log^2 t)$ probes in one update, and ours $\Theta(\log t)$ (Section~\ref{sec:witness});
\item a Rust implementation that reproduces these counts exactly, and measurements on six real vocabularies, where the weighted search keeps pace on ordinary text and is up to 2.5$\times$ faster on long runs of one character (Section~\ref{sec:eval}).
\end{itemize}

\paragraph{Related work.}
Other work handles growing or long text through different interfaces, such as finite lookahead for streaming~\citep{berglund2023formalizing}, finite-state formulations~\citep{cognetta2025fst}, bounded-delay streaming~\citep{mamouras2026streaming}, window re-encoding (TokTier;~\citealp{zhang2026toktier}) and parallel tokenization (LoPT;~\citealp{shao2026lopt}).  We improve the bound within Jiang and Gong's interface, which keeps the exact tokenization of every prefix.

\section{Background: exact incremental BPE}
\label{sec:background}

\paragraph{BPE and the last token.}
A dictionary is an alphabet, such as the 256 bytes, with a ranked list of merge rules $(x,y)\to xy$.  The rules are applied in rank order, and each rule merges the non-overlapping occurrences of $x$ followed by $y$ in one left-to-right pass.  Like Jiang and Gong, we assume a \emph{normalized} dictionary, in which no two rules produce the same token and tokenizing a token's spelling returns it.  For a token $v$ produced by $(x,y)\to v$, we write $\pre(v)=x$ and $\suc(v)=y$, which is a shorter suffix of $v$.  Let $\lastt(s)$ be the last token of the tokenization of $s$.  Given $\lastt$ of every prefix of $s$ and a new byte $c$, the task is to compute $\lastt(sc)$.

\paragraph{Candidates.}
The tokens that end $sc$ are all suffixes of the longest one, $\tau$, which an Aho--Corasick automaton~\citep{aho1975} finds in $O(1)$ time.  Linking each such \emph{candidate} $v\ne c$ to $\suc(v)$ gives the \emph{candidate tree} (Jiang and Gong's suffix-successor tree).  It is rooted at the byte $c$ and has at most $|\tau|\le t$ nodes.  We use three facts that Jiang and Gong prove (Appendix~\ref{app:framework}).
\begin{itemize}
\item[\textbf{F1}] \textbf{Valid path.}  Call $v$ \emph{valid} if it is a single byte, or if the text before $\suc(v)$, tokenized on its own, still ends in $\pre(v)$ right after the pass of $v$'s rule.  The valid candidates form a path from the root, and $\lastt(sc)$ is its deepest node~\citep[Thm.~4.2]{jiang2026incremental}.
\item[\textbf{F2}] \textbf{Constant-time tests.}  Number all tokens by a depth-first traversal.  Each $v$ then has an interval $I(v)$ such that $v$ is valid exactly when the stored last token of the text before $\suc(v)$ has its number in $I(v)$, and different children of a node have disjoint intervals~\citep[Sec.~4.3]{jiang2026incremental}.  At a valid node $u$, the child that continues the path is therefore the one whose interval contains a single number $x$, the number of the last token of the text before $u$.  We call finding this child a \emph{branch search}, and each interval it reads a \emph{probe}.
\item[\textbf{F3}] \textbf{Centroid search.}  The search keeps a connected set $C$ of candidates containing the answer and tests a \emph{centroid} $u$ of $C$, a node whose removal leaves pieces of at most $|C|/2$ nodes.  If $u$ is invalid, the search moves to the piece containing $u$'s parent.  Otherwise a branch search selects the piece below one child, or returns $u$ if no interval contains $x$~\citep[Sec.~5.3]{jiang2026incremental}.
\end{itemize}
Since every step at least halves $|C|$, there are at most $\lfloor\log_2 t\rfloor+1$ steps, and only the last can end with a failed branch search.  Jiang and Gong use ordinary binary search, which we call \emph{count-balanced}.  Over $d$ intervals it needs $\Theta(\log d)$ probes, and $d$ can be $\Theta(t)$, so an update costs $O(\log^2 t)$.

\begin{figure}[t]
\centering
\begin{tikzpicture}[
  tok/.style={font=\ttfamily\small, inner sep=1.4pt},
  sz/.style={font=\scriptsize, text=muted, fill=white, inner sep=0.6pt},
  thin/.style={draw=muted, line width=0.5pt},
  onpath/.style={draw=black, line width=1.3pt},
  stepno/.style={circle, draw=black, line width=0.4pt, font=\scriptsize, inner sep=0pt, minimum size=8.5pt},
  ivl/.style={draw=muted, line width=0.5pt, minimum width=1.15cm, minimum height=0.4cm, font=\ttfamily\small, inner sep=0pt},
]
  % (a) candidate tree for @abcdefg$
  \node[font=\small, anchor=west] at (-0.25,0.62) {(a) Candidates for \tk{@abcdefg\$}};
  \node[tok] (r)     at (2.9, 0)    {\$};
  \node[tok] (g)     at (0.3,-0.9)  {g\$};
  \node[tok] (fg)    at (1.3,-0.9)  {fg\$};
  \node[tok] (efg)   at (2.4,-0.9)  {efg\$};
  \node[tok] (defg)  at (3.9,-0.9)  {defg\$};
  \node[tok] (cdefg) at (2.95,-1.8) {cdefg\$};
  \node[tok, draw=black, rounded corners=2pt, inner sep=1.8pt] (bcdefg) at (4.85,-1.8) {bcdefg\$};
  \node[tok, text=muted] (abcdefg) at (4.85,-2.7) {abcdefg\$};
  \draw[thin]   (r) -- node[sz, pos=0.72] {1} (g);
  \draw[thin]   (r) -- node[sz, pos=0.72] {1} (fg);
  \draw[thin]   (r) -- node[sz, pos=0.72] {1} (efg);
  \draw[onpath] (r) -- node[sz, pos=0.62] {4} (defg);
  \draw[thin]   (defg) -- node[sz, pos=0.6] {1} (cdefg);
  \draw[onpath] (defg) -- node[sz, pos=0.55] {2} (bcdefg);
  \draw[thin, dashed] (bcdefg) -- node[sz, pos=0.5] {1} (abcdefg);
  \node[stepno, right=4pt of r] {1};
  \node[stepno, right=3pt of defg] {2};
  \node[stepno, right=3pt of bcdefg] {3};
  \node[font=\scriptsize, text=muted, anchor=west] at ($(abcdefg.east)+(0.05,0)$) {invalid};

  % (b) branch search at the root
  \begin{scope}[xshift=7.3cm]
  \node[font=\small, anchor=west] at (-0.25,0.62) {(b) Branch search at \tk{\$}};
  \node[ivl] (i1) at (0.75,0) {g\$};
  \node[ivl, right=0pt of i1] (i2) {fg\$};
  \node[ivl, right=0pt of i2] (i3) {efg\$};
  \node[ivl, right=0pt of i3, line width=1.1pt, draw=black] (i4) {defg\$};
  \foreach \i/\w in {i1/1,i2/1,i3/1,i4/4} \node[font=\scriptsize, text=muted, below=0.5pt of \i] {$w=\w$};

  % count-balanced (ordinary binary search) tree
  \node[tok] (c3) at (1.25,-1.2) {efg\$};
  \node[tok] (c2) at (0.55,-1.9) {fg\$};
  \node[tok, draw=countbal, line width=1.1pt, rounded corners=2pt, inner sep=1.8pt] (c4) at (1.95,-1.9) {defg\$};
  \node[tok] (c1) at (0.05,-2.6) {g\$};
  \draw[thin] (c3) -- (c2);
  \draw[thin] (c2) -- (c1);
  \draw[draw=countbal, line width=1.3pt] (c3) -- (c4);
  \node[font=\scriptsize, anchor=north] at (1.0,-2.82) {count-balanced: 2 probes};

  % component-weighted tree
  \node[tok, draw=compw, line width=1.1pt, rounded corners=2pt, inner sep=1.8pt] (w4) at (5.0,-1.2) {defg\$};
  \node[tok] (w2) at (4.25,-1.9) {fg\$};
  \node[tok] (w1) at (3.7,-2.6) {g\$};
  \node[tok] (w3) at (4.8,-2.6) {efg\$};
  \draw[thin] (w4) -- (w2);
  \draw[thin] (w2) -- (w1);
  \draw[thin] (w2) -- (w3);
  \node[font=\scriptsize, anchor=north] at (4.45,-2.82) {component-weighted: 1 probe};
  \end{scope}
\end{tikzpicture}
\caption{Running example (Section~\ref{sec:witness}, $k=3$) on the text \tk{@abcdefg\$}.  (a)~Candidates, each linked to its suffix $\suc(v)$, with subtree sizes on the edges.  Because \rul{@}{a} takes \tk{a}, \tk{abcdefg\$} is invalid, and the valid path (bold) ends at \tk{bcdefg\$}, the last token of \tk{@a}\,\textbar\,\tk{bcdefg\$}.  Circles number the centroids tested.  (b)~At the root, $x$ is the number of \tk{defg}, which lies in the last of four intervals, and $w$ is the size of the piece each interval leads to.  Binary search probes two intervals, the weighted tree one.}
\label{fig:example}
\end{figure}

\section{Component-weighted branch search}
\label{sec:method}

Consider a step at a valid centroid $u$ of $C$, with $m=|C|$.  Let $I_1,\dots,I_d$ be the intervals of $u$'s children in $C$, in increasing order, and let $C_i$ be the piece that the search enters when $x\in I_i$.  We give $I_i$ the weight $w_i=|C_i|$.  The pieces are disjoint subsets of $C\setminus\{u\}$, so $W=\sum_iw_i\le m-1$.  We store the intervals in a binary search tree that keeps their order and is balanced by weight.  Its root holds the first interval at which the running sum of weights reaches $W/2$, so each side weighs at most $W/2$, and both sides are built in the same way.  A query $x$ at a node holding $[L,R)$ goes left if $x<L$, goes right if $x\ge R$, and otherwise returns that interval.  If it reaches an empty subtree, no interval contains $x$.  The answers are the same as those of binary search, also when $x$ falls between intervals (Lemma~\ref{lem:equiv}).  Each probe that does not return moves to a side with at most half of the current weight, while a subtree that contains $I_i$ weighs at least $w_i$ and a nonempty one weighs at least $1$.  This gives the following bound.

\begin{lemma}[Local cost]
\label{lem:local}
A query that returns $I_i$ probes at most $1+\lceil\log_2(W/w_i)\rceil$ intervals, and a query that returns nothing probes at most $1+\lceil\log_2 W\rceil$, or none if there are no intervals.
\end{lemma}

\begin{theorem}[$O(\log t)$ updates]
\label{thm:main}
With component-weighted instead of count-balanced branch search, the algorithm of \citet{jiang2026incremental} returns the same token after every byte and takes $O(\log t)$ worst-case time per byte, $O(n\log t)$ for $n$ bytes.  Each stored interval gains $O(1)$ words.
\end{theorem}
\begin{proof}
Let $C_0\supset\cdots\supset C_r$ be the components visited, and $m_j=|C_j|$.  Then $m_0\le t$ and $m_{j+1}\le m_j/2$ (F3), so $r\le\log_2 t$.  A successful branch search at $C_j$ selects the interval of weight $m_{j+1}$, and the total weight is $W<m_j$.  By Lemma~\ref{lem:local}, it probes at most $1+\lceil\log_2(m_j/m_{j+1})\rceil\le2+\log_2 m_j-\log_2 m_{j+1}$ intervals.  The differences $\log_2 m_j-\log_2 m_{j+1}$ are nonnegative and telescope, so over all such steps they add up to at most $\log_2 m_0-\log_2 m_r\le\log_2 t$.  Only the last step can end with a failed branch search (F3), which costs at most $1+\lceil\log_2 t\rceil$ probes.  The validity test and history read of each step (F2) and the automaton step take $O(1)$ time~\citep{jiang2026incremental}.  An update therefore takes $O(r+\log t)=O(\log t)$ time.  The weighted tree gives the same answers as binary search, so every step makes the same decision as before, and the output is identical.  Building the tree over $d$ intervals takes $O(d\log d)$ time (Appendix~\ref{app:weighted}).
\end{proof}

\paragraph{What is new.}
Weight-balanced and biased search trees are classical~\citep{mehlhorn1975,bent1985biased}, usually with access frequencies as weights and a small average cost as the aim.  Our weights are the sizes of the pieces that the branches lead to.  They are fixed when the dictionary is compiled and do not depend on the input.  With them, each branch search costs at most a constant plus the logarithm of how much it shrinks the search, so successive costs telescope to a worst-case bound.

\section{A lower bound for count-balanced search}
\label{sec:witness}

We give normalized dictionaries on which count-balanced search makes $\Theta(\log^2 t)$ probes in one update, and ours $\Theta(\log t)$.  Each branch the search takes is the heaviest and has the last interval.

\paragraph{Shape.}
Let $T_0$ be one node.  The tree $T_j$ has a root whose children are a copy of $T_{j-1}$, called the \emph{heavy} child, and $2^{j-1}-1$ leaves, so $|T_j|=2^j$.  Jiang and Gong's implementation takes the top node of a component as its centroid unless one of its children holds strictly more than half of the component.  The heavy child holds exactly half, so the centroid of $T_j$ is its root.  Its branch weights are $1,\dots,1,2^{j-1}$, and the search continues in $T_{j-1}$.

\paragraph{Dictionary.}
Let $N=2^k$ and $S=a_0\cdots a_{N-2}\tk{\$}$, where the letters $a_p$ are distinct, \tk{\$} is an end marker and \tk{@} is a context symbol.  Cut $a_0\cdots a_{N-2}$ into blocks of lengths $1,2,4,\dots,2^{k-1}$.  The rules come in three groups, in decreasing priority.  \emph{Block rules} build each block into one token, from right to left.  The \emph{context rule} is $(\tk{@},a_0)\to\tk{@}a_0$.  \emph{Suffix rules}, from the shortest suffix to the longest, join the rest of each position's block to the suffix of $S$ that follows the block.  Writing $x\cdot y$ for $(x,y)\to xy$, the rules for $k=3$ are
\begin{align*}
&\text{blocks: }\rul{b}{c},\ \rul{f}{g},\ \rul{e}{fg},\ \rul{d}{efg};\qquad\text{context: }\rul{@}{a};\\
&\text{suffixes: }\rul{g}{\$},\ \rul{fg}{\$},\ \rul{efg}{\$},\ \rul{defg}{\$},\ \rul{c}{defg\$},\ \rul{bc}{defg\$},\ \rul{a}{bcdefg\$}.
\end{align*}

\begin{proposition}
\label{prop:witness}
For every $k\ge1$, this dictionary is normalized, and on the input $\tk{@}S$ the update for the final \tk{\$} follows the heavy path of $T_k$.  At each level $j\ge2$, the interval to be found is the last of $2^{j-1}$.  So count-balanced search makes $\Theta(k^2)$ probes and ours makes $k$, with $t=2^k$.  A variant over the fixed alphabet $\{\tk{0},\tk{1},\tk{[},\tk{]},\tk{\$},\tk{@}\}$ has $t=\Theta(k2^k)$, giving $\Theta(\log^2 t)$ against $\Theta(\log t)$ probes.
\end{proposition}

\paragraph{Why it works.}
The block rules come first and stay inside their blocks, so every block becomes one token.  Each suffix of $S$ is then made by its own suffix rule, which links it to the suffix after its block, and these links form $T_k$ (Figure~\ref{fig:example}a).  At the root of $T_j$, the children's intervals begin at the numbers of $a_{q-1}$, $a_{q-2}a_{q-1}$, \dots, up to the whole block.  Each of these tokens is the $\suc$-child of the one before it, so the heavy child comes last.  Finally, \tk{@} takes $a_0$, which makes $S$ invalid and changes nothing else, so the search runs down the heavy path.  Appendices~\ref{app:witness} and~\ref{app:macro} give the full proofs.

\section{Evaluation}
\label{sec:eval}

Both searches share one Rust implementation, which differs only in the branch search and counts probes exactly; the code is at \codeurl.  Count-balanced search ends with a separate probe of the right endpoint, which Figure~\ref{fig:example}b leaves out, and a weighted probe can compare $x$ with both endpoints.  With real vocabularies we tokenize continuous byte strings, without pre-tokenization, text normalization or special tokens.  Appendix~\ref{app:evaluation} gives the inputs and protocol.

\begin{wrapfigure}{r}{0.44\textwidth}
\vspace{-12pt}
\centering
\begin{tikzpicture}
\begin{axis}[
  width=\linewidth, height=3.3cm,
  xmin=3.6, xmax=13.4, ymin=0, ymax=112,
  xtick={4,7,10,13}, minor xtick={5,6,8,9,11,12},
  xticklabel={$2^{\pgfmathprintnumber{\tick}}$},
  ytick={0,25,50,75,100},
  xlabel={maximum token length $t$}, ylabel={probes},
  tick label style={font=\scriptsize}, label style={font=\scriptsize},
  xlabel style={yshift=3pt}, ylabel style={yshift=-4pt},
  axis line style={draw=muted}, tick style={draw=muted},
  ymajorgrids, grid style={draw=black!8},
  legend style={font=\scriptsize, draw=none, fill=none, at={(0.02,0.98)}, anchor=north west, row sep=-1pt},
  legend cell align=left,
]
\addplot[countbal, line width=1.1pt, mark=*, mark size=1.5pt]
  coordinates {(4,11) (5,16) (6,22) (7,29) (8,37) (9,46) (10,56) (11,67) (12,79) (13,92)};
\addlegendentry{count-balanced}
\addplot[compw, line width=1.1pt, dashed, mark=square*, mark size=1.3pt, mark options={solid}]
  coordinates {(4,4) (5,5) (6,6) (7,7) (8,8) (9,9) (10,10) (11,11) (12,12) (13,13)};
\addlegendentry{component-weighted}
\end{axis}
\end{tikzpicture}
\vspace{-6pt}
\caption{Probes on the final update, $t=2^k$.  Markers: measured; lines: $k(k+1)/2+1$ and $k$.}
\label{fig:scaling}
\vspace{-8pt}
\end{wrapfigure}
\paragraph{Constructed family.}
Every measured count equals the exact count, $k(k+1)/2+1$ for count-balanced and $k$ for component-weighted search (Figure~\ref{fig:scaling}).  At $t=8192$ this is 92 probes against 13.  For $k\le7$ our tests check every prefix against literal BPE.

\paragraph{Ordinary text.}
On six vocabularies (GPT-2, Pythia, Qwen2.5, DeepSeek-V4, OLMo-3.1 and BLOOM), with English prose, Python code and synthetic agent-style JSON, the weighted search makes 36--52\% fewer probes, and its largest count per byte falls from 4--17 to 2--6 (Table~\ref{tab:ordinary}).  These searches are short, so most of the difference is the final endpoint probe.  Median throughput is 0.90--1.05 times that of the count-balanced search, and the ranges overlap in every case.  The weighted trees add 0.7--1.1\% to the compiled tables.

\paragraph{Long runs of one character.}
The difference shows on repetitive text.  On 128\,KiB runs of hyphens, and of spaces for Pythia and Qwen2.5, the largest count per byte falls from 12--21 to 4--7, and throughput rises 1.3--2.5$\times$ (Table~\ref{tab:runs}).  Runs of the letter \tk{a} change little, and GPT-2 needs no branch search on spaces.  Such runs occur, for example, in indentation and separator lines, and the bound also covers inputs crafted to be slow.

\section{Conclusion}

Weighting each branch of Jiang and Gong's centroid search by the size of its subproblem lowers the worst-case time of exact incremental BPE from $O(\log^2 t)$ to $O(\log t)$ per byte, with the same output.  The $O(\log^2 t)$ bound of the original search is tight on valid dictionaries.  On real vocabularies the weighted search keeps pace on ordinary text and is up to 2.5$\times$ faster on long runs of one character.  Whether $o(\log t)$ time per byte is possible, and how much the bound matters inside serving systems, remain open.

\bibliographystyle{plainnat}
\bibliography{references}

@article{aho1975,
  author  = {Aho, Alfred V. and Corasick, Margaret J.},
  title   = {Efficient String Matching: An Aid to Bibliographic Search},
  journal = {Communications of the ACM},
  volume  = {18},
  number  = {6},
  pages   = {333--340},
  year    = {1975},
  doi     = {10.1145/360825.360855}
}

@article{berglund2023formalizing,
  author  = {Berglund, Martin and van der Merwe, Brink},
  title   = {Formalizing {BPE} Tokenization},
  journal = {Electronic Proceedings in Theoretical Computer Science},
  volume  = {388},
  pages   = {16--27},
  year    = {2023},
  doi     = {10.4204/EPTCS.388.4}
}

@article{bent1985biased,
  author  = {Bent, Samuel W. and Sleator, Daniel D. and Tarjan, Robert E.},
  title   = {Biased Search Trees},
  journal = {SIAM Journal on Computing},
  volume  = {14},
  number  = {3},
  pages   = {545--568},
  year    = {1985},
  doi     = {10.1137/0214041}
}

@article{debruijn1946,
  author  = {de Bruijn, Nicolaas Govert},
  title   = {A Combinatorial Problem},
  journal = {Proceedings of the Section of Sciences of the Koninklijke Nederlandse Akademie van Wetenschappen te Amsterdam},
  volume  = {49},
  pages   = {758--764},
  year    = {1946}
}

@inproceedings{jiang2026incremental,
  author        = {Jiang, Shenghu and Gong, Ruihao},
  title         = {Incremental {BPE} Tokenization},
  booktitle     = {Proceedings of the 43rd International Conference on Machine Learning},
  series        = {Proceedings of Machine Learning Research},
  volume        = {306},
  pages         = {52312--52342},
  publisher     = {PMLR},
  year          = {2026},
  note          = {Spotlight},
  eprint        = {2605.30813},
  archivePrefix = {arXiv},
  primaryClass  = {cs.CL}
}

@article{mehlhorn1975,
  author  = {Mehlhorn, Kurt},
  title   = {Nearly Optimal Binary Search Trees},
  journal = {Acta Informatica},
  volume  = {5},
  pages   = {287--295},
  year    = {1975},
  doi     = {10.1007/BF00264563}
}

@article{mamouras2026streaming,
  author    = {Mamouras, Konstantinos and Li, Angela W. and Yang, Yudi},
  title     = {An Efficient Algorithm for Streaming {BPE} Tokenization},
  journal   = {Proceedings of the ACM on Programming Languages},
  volume    = {10},
  number    = {PLDI},
  articleno = {252},
  pages     = {2085--2108},
  year      = {2026},
  month     = jun,
  doi       = {10.1145/3808330}
}

@misc{zhang2026toktier,
  author        = {Zhang, Zhenyu and Cao, Zhichao},
  title         = {{TokTier}: Exact Stateful {CPU+GPU} Tokenization for Agentic {LLM} Serving},
  year          = {2026},
  eprint        = {2607.29678},
  archivePrefix = {arXiv},
  primaryClass  = {cs.CL},
  note          = {arXiv:2607.29678}
}

@article{cognetta2025fst,
  author  = {Cognetta, Marco and Okazaki, Naoaki},
  title   = {Tokenization as Finite-State Transduction},
  journal = {Computational Linguistics},
  volume  = {51},
  number  = {4},
  pages   = {1119--1149},
  year    = {2025},
  doi     = {10.1162/coli.a.23}
}

@inproceedings{sennrich2016bpe,
  author    = {Sennrich, Rico and Haddow, Barry and Birch, Alexandra},
  title     = {Neural Machine Translation of Rare Words with Subword Units},
  booktitle = {Proceedings of the 54th Annual Meeting of the Association for Computational Linguistics (Volume 1: Long Papers)},
  address   = {Berlin, Germany},
  publisher = {Association for Computational Linguistics},
  pages     = {1715--1725},
  year      = {2016},
  month     = aug,
  doi       = {10.18653/v1/P16-1162}
}

@inproceedings{shao2026lopt,
  author    = {Shao, Wei and Zheng, Lingchao and Wang, Pengyu and Zheng, Peizhen and Li, Jun and Fan, Yuwei},
  title     = {{LoPT}: Lossless Parallel Tokenization Acceleration for Long Context Inference of Large Language Model},
  booktitle = {Proceedings of the 64th Annual Meeting of the Association for Computational Linguistics (Volume 1: Long Papers)},
  address   = {San Diego, California, United States},
  publisher = {Association for Computational Linguistics},
  pages     = {33107--33122},
  year      = {2026},
  month     = jul,
  doi       = {10.18653/v1/2026.acl-long.1529}
}

\clearpage
\raggedbottom
\appendix

\section{The inherited framework}
\label{app:framework}

This appendix states the parts of \citet{jiang2026incremental} that our results use, so that the dependence is explicit.  Proofs are in their paper.

\paragraph{Semantics and model.}
A dictionary is a finite alphabet with a ranked list of merge rules.  Rules apply in rank order, each in one left-to-right pass that merges non-overlapping occurrences.  For example, a self-pair rule $(x,x)$ applied to $xxx$ merges the first two symbols.  A token is \emph{canonical} if tokenizing its spelling returns it.  A dictionary is \emph{normalized} if every token has a unique spelling and is canonical, and every token that is not atomic is produced by exactly one rule (its canonical rule), whose operands are tokens created earlier~\citep[Sec.~3.3]{jiang2026incremental}.  Normalization preserves tokenization exactly~\citep[App.~C]{jiang2026incremental}, and \citet[App.~A]{jiang2026incremental} discuss how to make standard vocabularies proper.  Our results do not extend to arbitrary merge lists without this normalization.  Computation is on a word RAM in which automaton transitions and history reads take $O(1)$ time.  We write $O(\log t)$ for $O(1+\log t)$.

\paragraph{Output contract.}
After every appended byte the algorithm returns the last token of the exact BPE tokenization of the current text and stores it, with the position where it starts.  By prefix consistency~\citep[Lemma~3.1]{jiang2026incremental}, following these back-pointers recovers the tokenization of any prefix.  The full token sequence is not rewritten after each byte.  A history read $H(\ell)$ returns the stored last token of the current text without its final $\ell$ bytes, or the empty prefix.  Jiang and Gong's eager-output extension is unaffected by our change.

\paragraph{Successor forest and validity.}
For a non-atomic token $v$ with canonical rule $(x,y)\to v$, write $\pre(v)=x$ and $\suc(v)=y$.  The links $v\to\suc(v)$ form a forest whose roots are the atomic tokens.  Number its nodes in preorder, visiting the children of each node from lowest to highest rule priority, and let $\dfs(v)$ be the number of $v$.  Each subtree is then an interval of numbers.  For non-atomic $v$, the \emph{valid interval} is $I(v)=[\dfs(\pre(v)),R)$, where $R$ is the number of the first child of $\pre(v)$, in visiting order, whose rule has priority at least that of $v$, or the end of $\pre(v)$'s subtree if there is none~\citep[Sec.~4.3]{jiang2026incremental}.  A candidate $v$ is \emph{valid} for the current text if it is atomic, or if $\dfs(H(|\suc(v)|))\in I(v)$.  This is Jiang and Gong's Prefix Last-Token Condition~\citep[Def.~4.1]{jiang2026incremental}, and the informal reading in F1 follows their Section~4.1.

\paragraph{F1: monotonic path.}
Let $\tau$ be the longest token that is a suffix of the current text $s$.  The candidate tree is the subgraph of the forest induced by the tokens that are suffixes of $\tau$.  It is a tree rooted at the last byte of $s$.  A candidate is valid if and only if it lies on the path from this root to $\lastt(s)$~\citep[Thm.~4.2]{jiang2026incremental}.

\paragraph{F2: disjoint siblings.}
The valid intervals of distinct children of one node are disjoint~\citep[Sec.~4.3]{jiang2026incremental}.  Every child $v$ of $u$ has $\suc(v)=u$, so all children are tested with the same number $x=\dfs(H(|u|))$, and at a valid $u$ the child that continues the path is the unique child whose interval contains $x$, if any.

\paragraph{F3: centroid search.}
For every token $\tau$, the candidate tree of $\tau$ is decomposed recursively.  In a component, start at its top node and move to a child whose part of the component has strictly more than half of its nodes, as long as such a child exists.  The node reached is the centroid.  Removing it splits the component into an upward piece containing its parent, if present, and one downward piece per remaining child, each with at most half of the nodes.  At a centroid $u$ we store the intervals of the children of $u$ that lie in the current component.  A child outside it was removed at an earlier level as a centroid, and the search moved up from it only because it was invalid.  An update first looks up $\tau$ and then searches its candidate tree, which has at most $|\tau|\le t$ nodes.  At centroid $u$: if $u$ is invalid, continue in the upward piece; if $u$ is valid, continue in the piece of the child whose interval contains $x$, or return $u$ if there is none~\citep[Sec.~5.3]{jiang2026incremental}.  A failed branch search ends the search, so it happens at most once per update.

\section{Weighted branch search: details}
\label{app:weighted}

\begin{lemma}[Weighted pivot]
For positive weights $w_1,\dots,w_d$ with total $W$, let $p$ be the smallest index with $\sum_{i\le p}w_i\ge W/2$.  Then $\sum_{i<p}w_i<W/2$ and $\sum_{i>p}w_i\le W/2$.
\end{lemma}
\begin{proof}
The first inequality is the minimality of $p$, and the second is $\sum_{i>p}w_i=W-\sum_{i\le p}w_i\le W/2$.
\end{proof}

\paragraph{Construction.}
For an ordered list $(I_i,w_i)_{i=1}^d$, store the pivot interval at the root and recurse on the strict left and right sublists.  An empty sublist becomes a null link.  A query $x$ at $I_p=[L_p,R_p)$ goes left if $x<L_p$, goes right if $x\ge R_p$, and returns $p$ if $L_p\le x<R_p$.  A null link returns failure.  For $d=0$ the query fails without a probe.

\begin{lemma}[Same answers as binary search]
\label{lem:equiv}
The weighted tree returns $I_i$ exactly when $x\in I_i$, and failure exactly when $x$ lies in no stored interval.
\end{lemma}
\begin{proof}
The intervals are disjoint and sorted, so every interval to the left of $I_p$ ends at or before $L_p$, and every interval to the right of it starts at or after $R_p$.  Hence no comparison discards an interval that contains $x$, and induction over the sublists proves the claim, including for queries that fall in gaps.
\end{proof}

Lemma~\ref{lem:local} bounds the depth: if interval $i$ is not at the current node, the side containing it weighs at most half of the current total.  After $h-1$ such moves that total is at most $W/2^{h-1}$ but at least $w_i$.  For an unsuccessful query, the last subtree probed still has weight at least $1$.

\paragraph{Preprocessing and space.}
Store prefix sums of the weights.  For a sublist, binary search for the first prefix sum that reaches half of the sublist's total.  Organizing a sorted list of $d$ intervals takes $O(d\log d)$ time and $O(d)$ temporary space.  Our implementation adds two child links to each interval record, which adds 0.7--1.1\% to the compiled tables of the vocabularies in Appendix~\ref{app:evaluation}.

\paragraph{Short lists.}
Short lists can also be scanned.  Scanning at most a constant number $c$ of intervals costs $O(c)$ per step, so a hybrid that scans short lists and uses the weighted tree for longer ones keeps Theorem~\ref{thm:main}.

\section{The constructed family}
\label{app:witness}

\paragraph{Shape and indexing.}
Let $N=2^k$ and $r_j=2^j-1$ for $0\le j\le k$.  Number the nodes of $T_k$ in postorder as $0,\dots,N-1$.  For $1\le j\le k$, the nodes $r_{j-1},\dots,r_j-1$ are the children of $r_j$: node $r_{j-1}$ is the root of the heavy copy of $T_{j-1}$, which consists of nodes $0,\dots,r_{j-1}$, and the others are leaves.  Write $\pi(p)=r_j$ for $p\in[r_{j-1},r_j)$.  Node $r_k=N-1$ is the root.  Level $j$ owns the \emph{block} of positions $[r_{j-1},r_j)$, of length $2^{j-1}$.

\paragraph{String and dictionary.}
Let $a_0,\dots,a_{N-2}$ be letters such that all adjacent pairs $(a_p,a_{p+1})$ are distinct: for instance distinct letters or, in our byte-level instances, bytes $1,\dots,254$ taken from an order-two de Bruijn sequence~\citep{debruijn1946}.  The end marker \tk{\$} and the context symbol \tk{@} (bytes $0$ and $255$ in our instances) occur nowhere else.  Let
\[
  S=a_0a_1\cdots a_{N-2}\tk{\$},\qquad U_p=S[p..N),\qquad P_{p,q}=S[p..q).
\]
The alphabet is the set of letters together with \tk{\$} and \tk{@}.  The rules, in priority order, are:
\begin{enumerate}
\item \textbf{Block rules.} For each level $j$, with $[b,q)=[r_{j-1},r_j)$, and for $p=q-2,q-3,\dots,b$: $(a_p,P_{p+1,q})\to P_{p,q}$.  A block of length one needs no rule.
\item \textbf{Context rule.} $(\tk{@},a_0)\to\tk{@}a_0$.
\item \textbf{Suffix rules.} For $p=N-2,N-3,\dots,0$: $(P_{p,\pi(p)},U_{\pi(p)})\to U_p$.  The root $U_{N-1}=\tk{\$}$ is atomic.
\end{enumerate}

\paragraph{Position uniqueness.}
Because adjacent pairs are distinct, every string of length at least two occurs at most once in $a_0\cdots a_{N-2}$: two occurrences would repeat its first pair.

\begin{lemma}[Blocks]
\label{lem:blocks}
Each $P_{p,q}$ with $q=\pi(p)$ is canonical, and the block rules never combine letters from different blocks.
\end{lemma}
\begin{proof}
Within a block the rules build from right to left, so the right operand $P_{p+1,q}$ exists when the rule for $P_{p,q}$ is applied.  If the right operand is a single letter, the distinct pair $(a_p,a_{p+1})$ fixes the occurrence; if it is longer, position uniqueness fixes the right operand itself.  A block rule's output occurs only at its own position, so no rule applies across a block boundary.  No block contains \tk{\$} or \tk{@}, so no later rule applies inside a block.  Induction on decreasing $p$ shows that the tokenization of $P_{p,q}$ is $[P_{p,q}]$.
\end{proof}

\begin{lemma}[Suffixes and their links]
Each $U_p$ is canonical, and $\suc(U_p)=U_{\pi(p)}$.
\end{lemma}
\begin{proof}
Proceed by decreasing $p$.  In the tokenization of $U_p$, the block rules come first and, by Lemma~\ref{lem:blocks}, turn $S[p..\pi(p))$ into the single token $P_{p,\pi(p)}$ and every later block into one token.  The context rule does not apply, since \tk{@} does not occur.  A suffix rule for a position $h>p$ needs its right operand $U_{\pi(h)}$, which ends in \tk{\$}, at the end of the text, so its left operand must start at position $h$.  If $h<\pi(p)$, that position lies strictly inside $P_{p,\pi(p)}$, where no token starts; if $h\ge\pi(p)$, the rule acts only on the text after $P_{p,\pi(p)}$.  So no rule before the rule for $p$ involves $P_{p,\pi(p)}$, and by induction the rules turn $S[\pi(p)..N)$ into $U_{\pi(p)}$, as they do when it is tokenized alone.  Immediately before the rule for $p$, the tokens are therefore $P_{p,\pi(p)}$ and $U_{\pi(p)}$, and the rule produces $U_p$.  Rules for smaller positions come later and find a single token.  Thus $U_p$ is canonical with canonical rule $(P_{p,\pi(p)},U_{\pi(p)})$.  The argument uses positions rather than operands because, with repeated letters, another suffix rule can have the same left operand as the rule for $p$.
\end{proof}

\begin{lemma}[No other candidates]
The tokens that are suffixes of $S$ are exactly $U_0,\dots,U_{N-1}$, so the candidate tree of $S$ is $T_k$ with node $p$ represented by $U_p$.  The dictionary is normalized.
\end{lemma}
\begin{proof}
A suffix of $S$ ends in \tk{\$}.  Block tokens contain no \tk{\$}, and the context token begins with \tk{@}, which does not occur in $S$.  The remaining tokens are letters, \tk{\$} and the $U_p$.  The parent of $U_p$ is $U_{\pi(p)}$ by the previous lemma, which is the parent map of $T_k$.  All rule outputs have distinct spellings, every operand is created before it is used, and every token is canonical (Lemma~\ref{lem:blocks}, the previous lemma, and directly for $\tk{@}a_0$), so the dictionary is normalized.
\end{proof}

\begin{lemma}[Interval order]
At the centroid $U_q$ of level $j$, with block $[b,q)$, the intervals of the children appear in the order $U_{q-1},U_{q-2},\dots,U_b$, with weights $1,\dots,1,2^{j-1}$.
\end{lemma}
\begin{proof}
Write $B_p=P_{p,q}$, so $\pre(U_p)=B_p$ and $I(U_p)$ starts at $\dfs(B_p)$.  The block rule for $B_{p-1}$ has right operand $B_p$, so $B_{p-1}$ is a child of $B_p$ in the forest, and preorder gives $\dfs(B_{q-1})<\dfs(B_{q-2})<\cdots<\dfs(B_b)$.  The intervals are disjoint (F2), so they appear in this order.  The child $U_b$ roots the heavy copy of $T_{j-1}$, of size $2^{j-1}$, and every other child is a leaf; when the search reaches level $j$, the component is exactly this copy of $T_j$.
\end{proof}

\begin{lemma}[The forced update]
On the input $\tk{@}S$, the update for the final byte has $\tau=S=U_0$, which is invalid, and $\lastt(\tk{@}S)=U_1$.  The centroid search visits the roots of $T_k,T_{k-1},\dots,T_1$, selects the heavy (last) interval at each of $T_k,\dots,T_2$, and ends with a failed branch search at $T_1$.
\end{lemma}
\begin{proof}
The automaton reports $S$ as the longest token ending the text.  In $\tk{@}S$ the context rule, which precedes all suffix rules, merges \tk{@} with $a_0$, so the rule forming $U_0$ never finds $a_0$, and $U_0$ is invalid; this holds whether or not an implementation tests $\tau$ before searching.  The text after $\tk{@}a_0$ is untouched and forms $U_1$.

For $1\le j\le k$, removing the suffix $U_{r_j}$ from $\tk{@}S$ leaves $\tk{@}a_0\cdots a_{r_j-1}$, whose blocks are complete.  Its tokenization is $\tk{@}a_0$ followed by the blocks of levels $2,\dots,j$, so its last token is the block of level $j$, or $\tk{@}a_0$ when $j=1$.  Two consequences follow.  First, for $j<k$ the text before $\suc(U_{r_j})=U_{r_{j+1}}$ ends in the block of level $j+1$, which is $\pre(U_{r_j})$; hence $U_{r_j}$ is valid, and so is the atomic root $U_{r_k}$.  Second, at level $j\ge2$ the branch query is the number of the block of level $j$, which is $\pre(U_{r_{j-1}})$, so it lies in the heavy interval $I(U_{r_{j-1}})$.  At level $1$ the only child is $U_0$.  All children of $a_0$ in the forest, including $\tk{@}a_0$, come from rules of higher priority than the suffix rule of $U_0$, so $I(U_0)$ contains $\dfs(a_0)$ but not $\dfs(\tk{@}a_0)$.  The branch search fails and returns $U_1$.
\end{proof}

\paragraph{Exact counts.}
Our count-balanced implementation finds the last lower endpoint at most $x$ by binary search and then checks that interval's right endpoint as a separate probe.  At level $j\ge2$ the target is the last of $2^{j-1}$ intervals, which costs $j-1$ probes plus the endpoint check.  The final miss at level $1$ costs two.  Hence
\[
  Q_{\mathrm{count}}(k)=\sum_{j=2}^{k}j+2=\frac{k(k+1)}{2}+1,
  \qquad
  Q_{\mathrm{weighted}}(k)=k,
\]
since the weighted tree probes one interval at each of the $k-1$ successful levels and one at the final miss.  The lower-order terms depend on the implementation; the $\Theta(k^2)$ lower bound needs only that the target is the last of $2^{j-1}$ intervals at every level.

\paragraph{Worked example.}
For $k=3$ (Figure~\ref{fig:example}) the text is \tk{@abcdefg\$}.  The longest candidate $\tau=\tk{abcdefg\$}$ is invalid: the text before $\suc(\tau)=\tk{bcdefg\$}$ is \tk{@a}, which tokenizes to the single token \tk{@a} and so does not end in \tk{a}.  The search first tests the root \tk{\$}, which is atomic and hence valid.  Its query is the number of the last token of \tk{@abcdefg}, which tokenizes to \tk{@a}\,\textbar\,\tk{bc}\,\textbar\,\tk{defg}; so $x=\dfs(\tk{defg})$, which lies in $I(\tk{defg\$})$, the last of four intervals.  Next, \tk{defg\$} is valid, because the text before \tk{\$} ends in \tk{defg}.  Its query is the number of the last token of \tk{@abc}, namely \tk{bc}, which lies in $I(\tk{bcdefg\$})$, the later of two intervals.  Finally, \tk{bcdefg\$} is valid, because the text before \tk{defg\$} ends in \tk{bc}.  Its query is the number of the last token of \tk{@a}, which is \tk{@a} itself, and $I(\tk{abcdefg\$})=\{\dfs(\tk{a})\}$ does not contain it, because \tk{@a} is a child of \tk{a} formed by a rule of higher priority.  The search returns \tk{bcdefg\$}.  Our count-balanced search probes $3+2+2=7=k(k+1)/2+1$ intervals, and the weighted trees $1+1+1=3=k$.  At the root it probes three intervals rather than the two of Figure~\ref{fig:example}b, because it checks the right endpoint separately.

\paragraph{Independent check.}
We also checked the statements of this appendix and of Appendix~\ref{app:macro} with a short Python script (\texttt{scripts/check\_construction.py} in our repository), written separately from our Rust implementation, that implements literal BPE, the valid intervals, the strict-half centroid decomposition and both branch searches.  For distinct letters with $k\le6$, repeated letters whose adjacent pairs are distinct with $3\le k\le6$, and the fixed-alphabet variant with $k\le4$, it confirms that every token is canonical, that the candidate tree is $T_k$ with parent map $\pi$, that each centroid's intervals come in the stated order with weights $1,\dots,1,2^{j-1}$, that both searches agree with literal BPE on every prefix of the input, and the exact probe counts above.

\section{Fixed-alphabet encoding}
\label{app:macro}

Distinct pairs over 254 bytes limit the direct construction to a finite range.  For an unbounded family we use the alphabet $\{\tk{0},\tk{1},\tk{[},\tk{]},\tk{\$},\tk{@}\}$.  For $p\in\{0,\dots,N-2\}$, let
\[
  A_p=\tk{[}\,\bin_k(p)\,\tk{]},\qquad S'=A_0A_1\cdots A_{N-2}\tk{\$},
\]
where $\bin_k(p)$ is the $k$-bit binary code of $p$.  Rules come in two phases.  First, with highest priority, rules build every code suffix $y\tk{]}$ from right to left in order of increasing length, and then merge $(\tk{[},\bin_k(p)\tk{]})\to A_p$.  Equal code suffixes share one rule.  Second, the block, context and suffix rules of Appendix~\ref{app:witness}, with each $a_p$ replaced by the token $A_p$ and $a_0$ in the context rule replaced by $A_0$.

\begin{lemma}[Simulation]
In the tokenization of any substring of $\tk{@}S'$ that consists of whole codewords, the first phase turns every codeword into its token $A_p$.  The second phase then behaves exactly as the direct construction with letters $a_p$, and all statements of Appendix~\ref{app:witness} carry over.
\end{lemma}
\begin{proof}
Code-suffix rules are ordered by increasing length, so each right operand exists before the rule that extends it.  No first-phase rule has \tk{]} as its left operand or \tk{[} as its right operand, so none crosses a boundary $\tk{]}\tk{[}$ between codewords.  All codes have length $k$, so the token $\bin_k(p)\tk{]}$ starts right after its \tk{[}, and the next rule forms $A_p$.  The second phase has lower priority and therefore sees the codeword sequence $A_0A_1\cdots$, over which it is isomorphic to Appendix~\ref{app:witness}.  The $A_p$ are distinct, so adjacent pairs are distinct.  Intermediate code tokens end in \tk{]} and are not suffixes of the text, which ends in \tk{\$}, and \tk{@} occurs only in the context rule.  Every operand is canonical when used, so the dictionary is normalized.
\end{proof}

The longest token is $S'$, so $t=|S'|=(N-1)(k+2)+1=\Theta(k2^k)$ and $\log t=\Theta(k)$.  The $\Theta(k^2)$ against $\Theta(k)$ separation becomes $\Theta(\log^2t)$ against $\Theta(\log t)$, as stated in Proposition~\ref{prop:witness}.

\section{Evaluation details}
\label{app:evaluation}

\paragraph{Constructed family.}
Figure~\ref{fig:scaling} was measured for the first version of this paper.  The repository's \texttt{witness\_probes} example reproduces every point up to $t=4096$; its dense matcher needs about 16\,GB of memory for $t=8192$.  Our tests compare both searches with literal BPE on every prefix of the constructed inputs for $k\le7$ and check the exact probe counts for $k\le9$.  A probe reads one stored interval.

\paragraph{Real vocabularies.}
The real-vocabulary results are archived measurements from 30 September and 1 October 2026, taken with the code that our repository packages as a standalone crate; its search code is unchanged.  One release build (Rust 1.97.1, overflow checks on) ran on an Apple M4 under macOS 26.1, without control of CPU affinity or other load.  Each case ran in fresh processes, three for GPT-2, Pythia and Qwen2.5 and five for DeepSeek-V4, OLMo-3.1 and BLOOM.  Each process took seven calibrated samples of the byte-at-a-time append loop, with history reserved beforehand.  We report the median of the process medians and, in brackets, their range.  Compilation, reservation, input loading and reconstruction were outside the timed loop.  Probe counts come from a separate instrumented pass and are deterministic.

\paragraph{Inputs.}
English prose is \emph{Alice's Adventures in Wonderland} (Project Gutenberg \#11, 151,191 bytes) and, for DeepSeek-V4, OLMo-3.1 and BLOOM, also \emph{Pride and Prejudice} (\#1342, 738,046 bytes).  Python code is \texttt{Lib/asyncio/base\_events.py} from CPython 3.13.0 (79,451 bytes).  The agent-style JSON is a synthetic sequence of 1,024 tool-call records with queries in four languages (254,804 bytes).  The runs are 131,072 bytes of one repeated byte: \tk{a}, a space or a hyphen.  The vocabularies are those of GPT-2, Pythia (the GPT-NeoX-20B tokenizer), Qwen2.5, DeepSeek-V4-Pro, OLMo-3.1 and BLOOM, with all of their merge rules: 50,000, 50,009, 151,387, 127,741, 100,000 and 250,434 rules.

\paragraph{Correctness.}
In both sets of runs, the output of every process was compared at every prefix with a common reference history.  For each of the 27 pairs of vocabulary and text, the reference was also checked against Hugging Face tokenizers 0.22.1 at 91 to 95 prefixes, namely the first 32, powers of two, 48 chosen at random and the whole text, 2,500 in all.  All comparisons passed.

\begin{table}[H]
\centering
\caption{Long runs of one character (131,072 bytes).  Probes per byte, mean and maximum, as count-balanced $\to$ component-weighted; throughput in MiB/s, median [range] of process medians; speedup of the weighted search.}
\label{tab:runs}
\scriptsize
\setlength{\tabcolsep}{3pt}
\begin{tabular}{llccrrr}
\toprule
Vocabulary & Run & Mean probes & Max probes & Count-balanced & Weighted & Speedup\\
\midrule
GPT-2   & \tk{a}  & $2.25\to1.50$ & $5\to2$  & 146.3 [145.1, 146.7] & 153.8 [149.9, 156.4] & 1.05\\
GPT-2   & space   & $0\to0$       & $0\to0$  & 221.5 [204.2, 223.8] & 222.7 [218.0, 224.9] & 1.01\\
GPT-2   & hyphen  & $4.97\to2.78$ & $12\to4$ & 71.0 [70.5, 71.1]    & 95.0 [94.0, 95.5]    & 1.34\\
Pythia  & \tk{a}  & $2.00\to1.00$ & $5\to2$  & 141.4 [138.6, 142.1] & 145.0 [143.2, 148.6] & 1.03\\
Pythia  & space   & $9.08\to4.85$ & $20\to7$ & 37.0 [37.0, 37.2]    & 71.4 [69.7, 71.4]    & 1.93\\
Pythia  & hyphen  & $6.60\to3.94$ & $16\to6$ & 56.0 [52.7, 56.9]    & 78.4 [77.9, 78.6]    & 1.40\\
Qwen2.5 & \tk{a}  & $2.00\to1.00$ & $5\to2$  & 140.2 [137.0, 141.8] & 146.6 [143.8, 148.6] & 1.05\\
Qwen2.5 & space   & $10.91\to5.83$ & $21\to7$ & 23.2 [23.0, 23.3]   & 57.0 [55.9, 57.3]    & 2.46\\
Qwen2.5 & hyphen  & $5.03\to2.95$ & $12\to5$ & 67.1 [66.9, 67.1]    & 92.9 [90.6, 94.2]    & 1.38\\
\bottomrule
\end{tabular}
\end{table}

\begin{table}[H]
\centering
\caption{Ordinary text.  Columns as in Table~\ref{tab:runs}; the last column is the weighted search's throughput relative to the count-balanced search.}
\label{tab:ordinary}
\scriptsize
\setlength{\tabcolsep}{3pt}
\begin{tabular}{llccrrr}
\toprule
Vocabulary & Text & Mean probes & Max probes & Count-balanced & Weighted & Ratio\\
\midrule
GPT-2       & Alice    & $0.168\to0.102$ & $6\to3$  & 54.7 [52.1, 69.3]    & 51.2 [47.0, 69.6]    & 0.94\\
GPT-2       & Python   & $0.178\to0.102$ & $5\to3$  & 88.5 [86.2, 89.5]    & 79.9 [70.9, 89.0]    & 0.90\\
GPT-2       & JSON     & $0.231\to0.125$ & $4\to2$  & 114.6 [111.2, 119.9] & 114.5 [100.5, 122.1] & 1.00\\
Pythia      & Alice    & $0.161\to0.100$ & $6\to3$  & 73.9 [73.1, 76.6]    & 72.8 [71.6, 77.1]    & 0.99\\
Pythia      & Python   & $0.315\to0.154$ & $14\to6$ & 88.9 [87.3, 90.2]    & 91.0 [88.7, 91.3]    & 1.02\\
Pythia      & JSON     & $0.346\to0.165$ & $5\to2$  & 117.1 [114.8, 119.7] & 119.5 [116.3, 121.5] & 1.02\\
Qwen2.5     & Alice    & $0.165\to0.102$ & $6\to3$  & 63.1 [54.0, 66.5]    & 60.5 [49.6, 68.2]    & 0.96\\
Qwen2.5     & Python   & $0.397\to0.204$ & $6\to3$  & 76.3 [73.8, 83.1]    & 79.8 [79.3, 81.3]    & 1.05\\
Qwen2.5     & JSON     & $0.179\to0.099$ & $4\to2$  & 135.7 [129.5, 136.9] & 135.3 [134.5, 140.9] & 1.00\\
DeepSeek-V4 & Alice    & $0.160\to0.101$ & $6\to3$  & 68.9 [50.2, 71.2]    & 66.8 [62.4, 67.3]    & 0.97\\
DeepSeek-V4 & Pride    & $0.214\to0.136$ & $15\to6$ & 53.6 [46.1, 54.3]    & 53.4 [46.3, 55.5]    & 1.00\\
DeepSeek-V4 & Python   & $0.397\to0.204$ & $7\to3$  & 79.9 [77.2, 82.5]    & 80.8 [54.5, 81.7]    & 1.01\\
OLMo-3.1    & Alice    & $0.165\to0.102$ & $6\to3$  & 64.3 [46.4, 70.1]    & 61.5 [47.3, 69.4]    & 0.96\\
OLMo-3.1    & Pride    & $0.215\to0.134$ & $7\to3$  & 55.4 [49.3, 55.6]    & 54.3 [42.4, 55.9]    & 0.98\\
OLMo-3.1    & Python   & $0.397\to0.204$ & $6\to3$  & 80.6 [77.1, 82.2]    & 81.8 [78.3, 82.2]    & 1.01\\
BLOOM       & Alice    & $0.154\to0.095$ & $7\to4$  & 72.6 [68.6, 74.6]    & 67.2 [59.7, 70.9]    & 0.93\\
BLOOM       & Pride    & $0.207\to0.128$ & $17\to6$ & 55.9 [45.0, 58.4]    & 56.3 [50.9, 57.0]    & 1.01\\
BLOOM       & Python   & $0.485\to0.254$ & $15\to6$ & 76.6 [73.9, 79.0]    & 78.3 [77.1, 80.2]    & 1.02\\
\bottomrule
\end{tabular}
\end{table}

\begin{table}[H]
\centering
\caption{Compilation time, excluding dictionary validation, and memory, from the processes for \emph{Alice}.  Charged tables are the compiler's accounting of table payload, most of it the dense matcher that both searches share; peak RSS covers the whole process.  Each pair is count-balanced, then weighted.}
\label{tab:memory}
\scriptsize
\setlength{\tabcolsep}{4pt}
\begin{tabular}{lrrrrrr}
\toprule
& \multicolumn{2}{c}{Compilation (ms)} & \multicolumn{2}{c}{Charged tables (MiB)} & \multicolumn{2}{c}{Peak RSS (MiB)}\\
Vocabulary & Count-bal. & Weighted & Count-bal. & Weighted & Count-bal. & Weighted\\
\midrule
GPT-2       & 136.1 & 164.7 & 214.9  & 217.0  & 257.4  & 260.7\\
Pythia      & 119.8 & 134.1 & 199.3  & 201.4  & 241.4  & 247.0\\
Qwen2.5     & 426.6 & 487.0 & 720.5  & 726.1  & 832.7  & 832.5\\
DeepSeek-V4 & 330.3 & 367.4 & 592.4  & 596.8  & 690.6  & 690.2\\
OLMo-3.1    & 268.6 & 291.6 & 473.4  & 477.9  & 548.2  & 557.4\\
BLOOM       & 781.3 & 888.6 & 1300.2 & 1309.7 & 1491.3 & 1500.5\\
\bottomrule
\end{tabular}
\end{table}

\end{document}